\documentclass[journal,draftclsnofoot,onecolumn]{IEEEtran}
\usepackage{amsmath,amsfonts,amssymb,cite,color,enumerate,ifpdf,array,mathtools}
\usepackage{enumerate}
\usepackage{amsthm}
\usepackage{dsfont}
\usepackage[noend]{algpseudocode}
\usepackage{bbm}
\usepackage{soul}

\algnewcommand{\IIf}[1]{\State\algorithmicif\ #1\ \algorithmicthen}
\algnewcommand{\EndIIf}{\unskip\ \algorithmicend\ \algorithmicif}
\algrenewcommand\algorithmicindent{1em}%

\theoremstyle{definition}

\newtheorem{Prop}{Proposition}

\newif\ifdoublecolumn
\doublecolumnfalse

\ifdoublecolumn
\def\TableOneShortWidth{0.61in}
\def\TableOneLongWidth{2.56in}

\def\FigScale{0.6}
\def\FigScaleS{0.4}
\def\Extra{0.5ex}

\def\TableMargin{-16pt}

\else

\def\TableOneShortWidth{0.7in}
\def\TableOneLongWidth{5.4in}

\def\FigScale{0.65}
\def\FigScaleS{0.45}
\def\Extra{0.5ex}

\def\TableMargin{-15pt}

\fi

\begin{document}

\title{A Sensing Contribution-based Two-layer Game for Channel Selection and Spectrum Access in Cognitive Radio Ad-hoc Networks}

\author{Yuan~Lu,
        Alexandra~Duel-Hallen,~\IEEEmembership{Fellow,~IEEE}%
				\thanks{This research was partially supported by the NSF grant CNS-1018447. This paper was presented in part at the 49th
Conference on Information Sciences and Systems, Baltimore, Maryland, Mar. 2015.}
\thanks{Y. Lu is with Huawei Technologies, 400 Crossing Blvd., Bridgewater, NJ, 08807 USA (e-mail: ylu8@ncsu.edu).}
\thanks{A. Duel-Hallen is with the Department of Electrical and Computer Engineering, North Carolina State University, Raleigh,
NC, 27606 USA  (e-mail: sasha@ncsu.edu).}
}
\maketitle

\begin{abstract}
In cognitive radio (CR) networks, the secondary users (SUs) sense the spectrum licensed to the primary users (PUs) to identify and possibly transmit over temporarily unoccupied channels. Cooperative sensing was proposed to improve the sensing accuracy, but in heterogeneous scenarios SUs do not contribute equally to the cooperative sensing result because they experience different received PU signal quality at their sensors. In this paper, a two-layer coalitional game is developed for distributed sensing and access in multichannel CR ad hoc networks where the SUs' transmission opportunities are commensurate with their sensing contributions, thus  fostering cooperation and eliminating free-riders. Numerical results show that the proposed two-layer game is computationally efficient and outperforms previously investigated collaborative sensing and spectrum access approaches for heterogeneous multichannel CR networks in terms of energy efficiency, throughput, SU fairness, and complexity. Moreover, it is demonstrated that this game is robust to changes in the network topology and the number of SUs in low-mobility scenarios. Finally, we propose a new physical-layer approach to distributing the network-level miss-detection (MD) constraints fairly among the interfering SUs for guaranteed PU protection and demonstrate the performance advantages of the AND-rule combining of spectrum sensing results for heterogeneous SUs. 
\end{abstract}


\begin{table}[t]
\caption{Main Notation}  
\label{Table:Notation} 
\centering 
\begin{tabular}{@{\extracolsep{-\Extra}}|>{\centering\arraybackslash}m{\TableOneShortWidth}|>{\arraybackslash}m{\TableOneLongWidth}|}
\hline\hline 
\textbf{Notation} & \multicolumn{1}{c|}{\textbf{Explanation}}\\
\hline\hline
$\mathcal{M}$, $\mathcal{N}$ &  Set of all SUs and set of all channels. \\
\hline
$C^n=(S,n)$ & {A top-layer coalition $C^n$ is a two-tuple, where $S\subseteq\mathcal{M}$ is a set of SUs and $n$ is the operating channel. } \\ 
\hline
$\mathcal{P}$ &  {A top-layer partition of $\mathcal{M}$ defines a set of disjoint top-layer coalitions $\mathcal{P}=\{C^1,C^2, \ldots ,C^N\}$, where $\cup_{n=1}^NC^n(1)=\mathcal{M}$.}\\
\hline
$x^{mC^n}$ & {Top-layer utility, given by the expected data rate of SU $m$ operating on channel $n$.}\\
\hline
$a^{mC^n}$ & {Bottom-layer game payoff, which measures the $m$th SU's priority for transmission over channel $n$.}\\
\hline
$R^{mn}$ & {Transmission rate of SU pair $m$ on channel $n$.}\\ 
\hline
$\gamma^{mn}$ & {SU-to-SU signal-to-noise ratio (SNR) of SU pair $m$ on channel $n$.}\\
\hline
{$B^n$, $\beta^n$} & {Bandwidth and availability probability of channel $n$.}\\
\hline
$(S,U^n)$ & {Bottom-layer game among a set of SUs $S$ on channel $n$ with the value function $U^n$.} \\
\hline
$\eta$ & {A bottom-layer coalition $\eta$ is a set of cooperative SUs sensing the same channel}. \\
\hline
$\rho^n$ & {Bottom-layer partition on channel $n$ defines a set of disjoint bottom-layer coalitions $\rho^n=\{\eta_1,\eta_2,\ldots,\eta_{|\rho^n|}\}$, where $\cup_{i=1}^{|\rho^n|} \eta_i=C^n(1)$.} \\
\hline
$\rho^n\backslash\{\eta\}$ & Relative complement of $\{\eta\}$ in $\rho^n$, i.e., the set of all bottom-layer coalitions in $\rho^n$ except $\eta$. \\
\hline
$U^n(\eta;\rho^n)$ & Value of bottom-layer coalition $\eta$, as given by $\eta$'s overall successful transmission probability on channel $n$ under the bottom-layer partition $\rho^n$ where $\eta\in \rho^n$. \\
\hline
\ifdoublecolumn
$P_\mathrm{MD}^n(\eta)$,\vspace{0.5ex}\newline $P_\mathrm{FA}^n(\eta)$\vspace{0.5ex}  & {Cooperative miss detection (MD) and false alarm (FA) probabilities of a bottom-layer coalition $\eta$ regarding the PU presence.} \\
\else
$P_\mathrm{MD}^n(\eta)$, $P_\mathrm{FA}^n(\eta)$& {Cooperative miss-detection (MD) and false alarm (FA) probabilities of a bottom-layer coalition $\eta$ regarding the PU presence.} \\
\fi
\hline 
$P_\mathrm{MD}^\mathrm{Ch}$ & {Integrated MD probability constraint on each channel.} \\
\hline
\ifdoublecolumn
$P_\mathrm{MD}^n(m)$,\vspace{0.5ex}\newline $P_\mathrm{FA}^n(m)$\vspace{0.5ex} & {Individual MD constraint and FA probability of SU $m$ on channel $n$.} \\
\else
$P_\mathrm{MD}^n(m)$, $P_\mathrm{FA}^n(m)$ & {Individual MD constraint and FA probability of SU $m$ on channel $n$.} \\
\fi
\hline 
{$\lambda^{mn}$} & { PU-to-SU SNR per sample at SU $m$ on channel $n$.} \\
\hline  
{$\nu$} & {Number of collected samples for spectrum sensing.} \\
\hline

\end{tabular} 
\vspace*{\TableMargin}
\end{table}

\section{Introduction}
Cooperative sensing exploits spatial diversity to improve sensing accuracy in cognitive radio (CR) systems\cite{CoopSenseSurvey}. While most investigations assume a fixed number of fully cooperative secondary users (SUs) with identical sensing capabilities monitoring a single channel, in practice, there are many possible channels for sensing and transmission, and the sensing accuracy varies over the spectrum and among the SUs. In ad-hoc networks, which do not utilize a central controller, and under the hardware constraints, SUs need to choose both the channels to sense and their collaborators\footnote{We consider only cooperation among the SUs. The PU-SU cooperation \cite{MatchingTheory} is out of the scope of this paper.} for spectrum sensing in a distributed manner. Moreover, since simultaneous SU transmission attempts cause mutual interference and/or collisions, resulting in wasted sensing effort and transmission opportunities, an SUs' agreement on sharing detected spectrum holes is also desirable.

Game theory has been utilized recently to model and analyze SU interactions in cooperative sensing \cite{PartitionFormSenseGame,RLiuEvoGame,ZHanCoalSenseGame,MultiChCoalSenseGame,FairPayoff,AuctionBased,Overlapp,HedonicSenseGame} and opportunistic access \cite{HedonicSenseGame,NetwEcon,JWangChSelSurvey}, but, to the best of our knowledge, only the game in \cite{HedonicSenseGame} takes into account transmission (spectrum access) opportunities when making channel sensing and collaboration decisions in a multichannel CR network. However, the game in \cite{HedonicSenseGame} is not suitable for a heterogeneous environment where the contributions of different SUs within a coalition can vary significantly, although such heterogeneity is typical in wireless networks. For example, suppose several SUs sense the same channel and attempt to transmit over it when it is sensed idle using a distributed medium access control (MAC) scheme. If some of these SUs are located in the proximity of the primary user (PU), they are likely to obtain more accurate sensing results, which lead to better transmission chances than for more distant SUs. In such heterogeneous scenarios, the SUs would benefit if the distributed MAC method was replaced by the following agreement on coordinated medium access: SUs with good sensing accuracy would share their sensing results with other SUs in exchange for a higher payoff (e.g., a higher chance or a larger time share of transmission) when a spectrum opportunity is detected. This observation motivates the proposed two-layer game where coalition formation across and within the channels is based on such an agreement on payoff allocation. In contrast, the game of \cite{HedonicSenseGame} forces all SUs sensing the same channel to cooperate altruistically and awards them equally. As a result, selfish SUs might choose to become free-riders \cite{RLiuEvoGame} by taking advantage of other SUs' accurate sensing results without contributing appropriately. This approach undermines fairness among individual SUs and, thus, might discourage SUs from contributing to cooperative sensing, causing the overall network throughput to decline. While a single-layer coalitional game could be devised to meet the proposed fairness and performance objectives, we demonstrate that the two-layer game structure is more attractive in distributed sensing and access due to reduced complexity, delay, and overhead requirements.

Moreover, to facilitate efficient spectrum sensing, we develop constraints on miss-detection (MD) rates of coalition members and provide novel insights into the performance of fusion rules of sensing results for heterogeneous networks under the constant detection rate (CDR) \cite{YCLiangWCNC} constraints.

The \textit{contributions} of this paper are:
\begin{itemize}
	\item Development and analysis of a two-layer coalitional game that includes:
	\begin{enumerate}[(i)]
		\item An efficient, stable, and distributed coalition formation algorithm that assigns SUs to channels within the CR spectrum.
		\item A contribution-based payoff allocation scheme to promote individual incentives for cooperation.
	\end{enumerate}
  \item Improved spectrum sensing approaches:
	\begin{enumerate}[(i)]
		\item A fair distribution method of the integrated network-level primary collision probability constraint among the coalitions and member SUs sharing the same channel.
		\item Demonstration of performance advantages of the AND-rule combining for heterogeneous sensing environments under the CDR constraints. 
	\end{enumerate}
\end{itemize}
While we have first proposed the two-layer game in \cite{YLuCISS}, that paper did not contain several key proofs, practical validation, and complexity analysis of the coalitional game and did not address sensing data fusion rules for heterogeneous environments.

The rest of this paper is organized as follows. Table~\ref{Table:Notation} summarizes significant notation. In section~\ref{Sec:Model}, we introduce the system model and formulate the proposed two-layer game, which is then analyzed in Section~\ref{Sec:Game} using coalitional game theory. Simulation results and comparison with \cite{HedonicSenseGame} are presented in Section~\ref{Sec:Sim}, and conclusions are drawn in Section~\ref{Sec:Conclusion}.

\section{System Model and Two-Layer Game Formulation}\label{Sec:Model}
We consider an overlay slotted\footnote{\label{Note:continuousPUtraffic}Our results can be extended to continuous PU traffic by adjusting the primary collision constraints\cite{JWangChSelSurvey}.} CR ad hoc network with multiple SU pairs and multiple channels under i.i.d. Bernoulli PU traffic. An SU can sense and access only one channel at each time slot due to the hardware constraints. The bidirectional interaction and information exchange within the proposed two-layer game framework is illustrated in Fig.~\ref{Fig:TwoLayerGame}: The top-layer game partitions the SUs into $N$ disjoint groups for sensing different channels, based on the expected payoff computed from the bottom layer. At the bottom-layer, a total number of N coalitional games are played in parallel, and the goal is to determine the collaborative structure among the SUs sensing a given channel as well as the payoff allocation. The two-layer game is introduced below and analyzed in Sections~\ref{Sec:Game} and~\ref{Sec:Sim} under several simplified and typically employed assumptions, e.g., negligible control overhead and time-invariant additive white Gaussian noise (AWGN) channels. We discuss the extension of the proposed game to settings where these assumptions are relaxed in Section~\ref{SubSec:Complexity}.
\subsection{The Top-Layer Game}\label{SubSec:TopDef}
\begin{figure}[!t]
    \centering
    \includegraphics[scale=\FigScaleS]{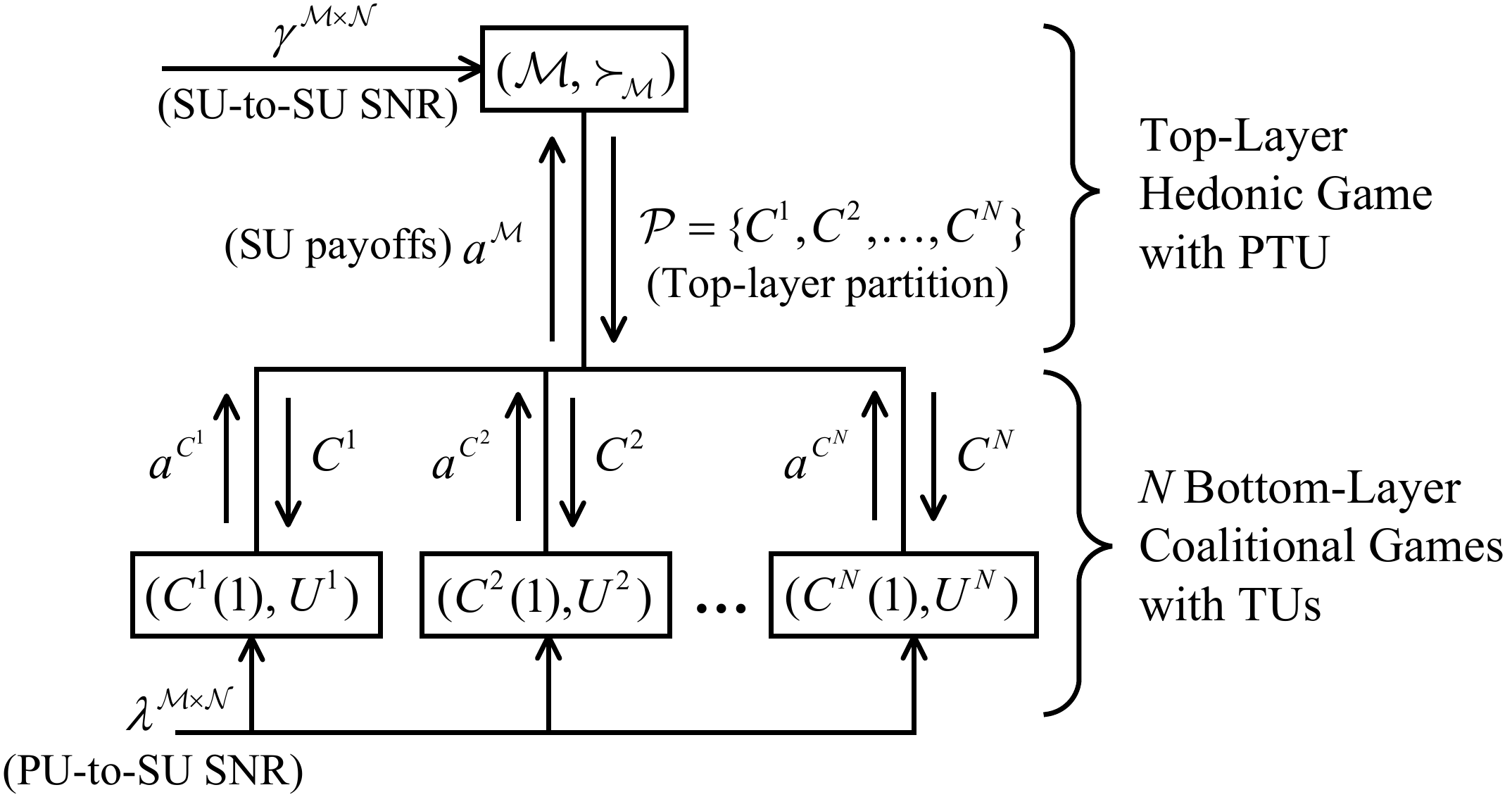}
    \caption{Two-layer coalitional game.}
    \label{Fig:TwoLayerGame}
\end{figure}

The proposed top-layer game is a \emph{hedonic game}, which is characterized by a set of players and a \emph{preference relation}\cite{Hedonic}. The set of players consists of all SUs $\mathcal{M}=\{1,2,\ldots,M\}$ over the set of all channels $\mathcal{N}=\{1,2,\ldots,N\}$, where each SU is associated with some coalition $C^n=(S,n)$\footnote{Note that a top-layer coalition is a two-tuple specifying both the set of SUs $S$ and the channel index $n$ because the same set of SUs $S$ can achieve different throughputs on different channels.}. The SUs in $S=C^n(1)$ sense and transmit over the $n$th channel, i.e., the sensing decision of each SU $m \in S$ is
\begin{equation}\label{Eq:SenseDec}
n^m_*=n\Leftrightarrow m\in C^n(1).
\end{equation}
The utility of the SU $m\in S$ of the top-layer coalition $C^n=(S,n)$ is given by
\begin{equation}\label{Eq:PTU}
x^{mC^n}=a^{mC^n}R^{mn}
\end{equation}
where $a^{mC^n}$ is the allocated individual payoff of SU $m$ provided by the bottom-layer game,  which is a transferable utility (TU) \cite{CoalGameSurvey}, and $R^{mn}$ is the data rate of SU $m$ on channel $n$ given by $R^{mn}=B^{n}\log_2(1+\gamma^{mn})$, a nontransferable utility (NTU) \cite{CoalGameSurvey}, which can vary greatly across different SU pairs due to their spatial separation. Thus, we refer to the utility (\ref{Eq:PTU}) as \emph{partially transferable utility (PTU).}

The utility (\ref{Eq:PTU}) measures the expected data rate of the $m$th SU over channel $n$. Moreover, following the hedonic game formulation\cite{Hedonic}, we define a preference relation that guides an SU when choosing a top-layer coalition. An SU $m$ prefers to move from channel $n$ to $\tilde{n}$ if the following preference relation is satisfied:
\begin{equation}\label{Eq:Preference}
\begin{split}
\tilde{C}^{\tilde{n}}\succ_m C^n\!\Leftrightarrow\! \left\{\begin{matrix}
 x^{m \tilde{C}^{\tilde{n}}}>x^{m C^n} \\ 
  \smashoperator[r]{\sum_{ i{\in}  \tilde{C}^n(1)}}a^{i{ \tilde{C}^n}}{+}\smashoperator{\sum_{i\in \tilde{C}^{\tilde{n}}(1)}}a^{i \tilde{C}^{\tilde{n}}}{ >}\smashoperator{\sum_{i\in C^n(1)}}a^{iC^n}{+}\smashoperator{\sum_{i\in C^{\tilde{n}}(1)}}a^{iC^{\tilde{n}}}
\end{matrix}\right.\raisetag{2.8\baselineskip}
\end{split}
\end{equation}
where $\tilde{C}^n{=}(C^n(1)\backslash \{m\},n)$ and $\tilde{C}^{\tilde{n}}$ are the top-layer coalitions that form on channels $n$ and $\tilde{n}$, respectively, if SU $m$ makes the move, and $C^n$ and $C^{\tilde{n}}{=}(\tilde{C}^{\tilde{n}}(1)\backslash\{m\},\tilde{n})$ are the existing top-layer coalitions on these channels. Thus, the top-layer coalition $\tilde{C}^{\tilde{n}}$ is preferable to $C^n$ for SU $m$ if (i) its expected data rate (\ref{Eq:PTU}) improves and (ii) the combined (or equivalently, the average) payoffs of all SUs on channels $n$ and $\tilde{n}$ improve. This \emph{social} utility improvement requirement (ii) introduces a loose constraint on selfish and short-sighted individual movements, with the aim to facilitate fast convergence of the partition formation process. In this sense, it is similar to, yet more effective than, the use of ``history sets'' in \cite{HedonicSenseGame} as will be demonstrated in Section~\ref{Sec:Sim}. For example, in Fig.~\ref{Fig:Example}, SU $m=4$ prefers to switch from channel $1$ to channel $2$ if this movement improves not only its individual throughput, but also the combined transmission opportunities (or equivalently, the average payoff) of all SUs currently residing on these channels, i.e., SUs $1$--$4$ and $6$.

Given $(\mathcal{M},\mathcal{N},a^{\mathcal{M}}, \gamma^{\mathcal{M}\times \mathcal{N}})$, the objective of the top-layer game is to form a \emph{Nash-stable} \cite{Hedonic} partition $\mathcal{P}$ (see Fig.~\ref{Fig:TwoLayerGame} and Table~\ref{Table:Notation}), where each player cannot unilaterally improve its utility $x^{m C^n}$ given Nash utilities of other players.

\begin{figure}[!t]
    \centering
    \includegraphics[scale=0.43]{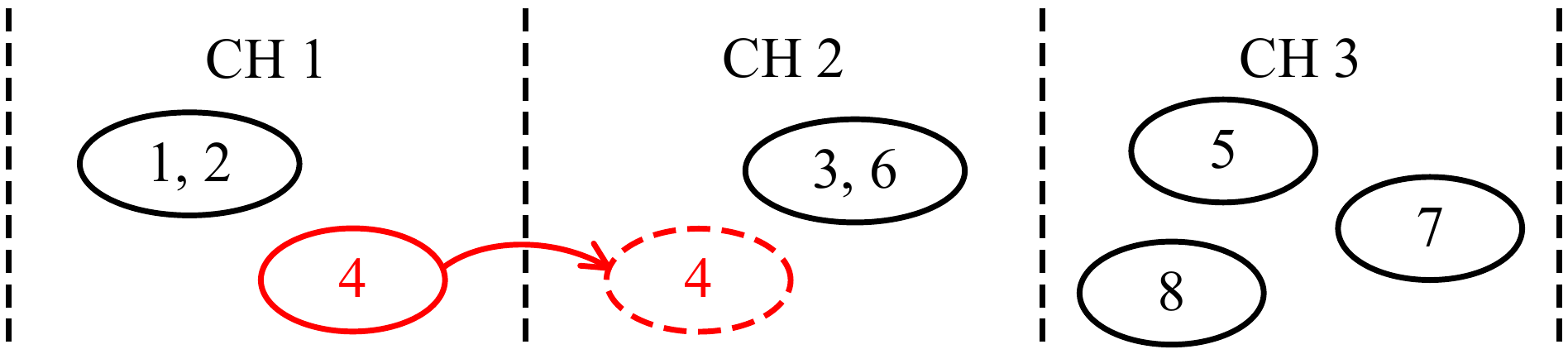}
    \caption{An example of the two-layer coalition structure and SU movements  with $\mathcal{M}=\{1,\ldots,8\}$ and $\mathcal{N}=\{1,2,3\}$. The top-layer partition is ${\mathcal P} = \{ (\{ 1,2,4\} ,1),(\{ 3,6\} ,2),(\{ 5,7,8\} ,3)\}$ and the bottom-layer partitions are ${\rho ^1} = \{ \{ 1,2\} ,\{ 4\} \}$, ${\rho ^2} = \{ \{ 3,6\} \} $ and ${\rho ^3} = \{ \{ 5\} ,\{ 7\} ,\{ 8\} \}$. SU $m=4$ prefers to move from channel $n=1$ to $\tilde{n}=2$ if (\ref{Eq:Preference}) is satisfied with $C^n=(\{1,2,4\},1)$, $C^{\tilde n}=(\{3,6\},2)$,  $\tilde{C}^n=(\{1,2\},1)$ and $\tilde{C}^{\tilde{n}}=(\{3,6,4\},2)$.}
    \label{Fig:Example}
\end{figure}

\subsection{The Bottom-Layer Game}\label{SubSec:BottomDef}
An output partition $\mathcal{P}$ of the top-layer game determines the set of SUs that can sense and transmit over each channel. The SUs sensing the same channel $n$ can further form disjoint bottom-layer coalitions, resulting in a bottom-layer partition $\rho^n$, as illustrated in Fig.~\ref{Fig:Example}. Within each bottom-layer coalition, the SUs exchange and combine their sensing results to improve the overall successful transmission probability of this coalition. The detected spectrum opportunities are then shared among the bottom-layer coalition members in a coordinated manner using a payoff allocation rule, which determines an SU's share of the slot for transmission and provides collision avoidance (cf. Section~\ref{SubSec:BottomSol}). Formally, we define a \emph{coalitional game} $(S,U^n)$ \cite{CoalGameSurvey} for each channel $n$ with
\begin{itemize}
	\item  \emph{a set of players}: the set of SUs $S\subset \mathcal{M}$ sensing channel $n$, i.e., $S=C^n(1)$ obtained from the top-layer game;
  \item  \emph{a value function}: a function $U^n$ that maps any subset of $S$ to a real value.
\end{itemize}

\emph{Within }each bottom-layer coalition, SUs coordinate channel access through negotiation, instead of relying on a distributed MAC scheme. In contrast, MAC is needed \emph{across} multiple bottom-layer coalitions since several such coalitions could simultaneously detect a spectrum opportunity correctly and, therefore, would compete for access. As in \cite{QZhaoTDFS}, we consider the following two MAC options:

\begin{enumerate}
  \item \emph{0/X-model}: The competing SUs do not employ a distributed MAC scheme. Therefore, their transmission collide and all competing SUs fail to transmit successfully.
  \item \emph{1/X-model}: All competing SUs resort to an ideal MAC scheme to obtain equal transmission chances\footnote{We assume negligible control overhead in this random access MAC model.}. After gaining the right to access, the winning SU shares its transmission opportunity with its bottom-layer coalition's members as required by their binding agreement on payoff allocation.
\end{enumerate}
The chosen MAC option affects the players' \emph{noncooperation cost}, which measures the impact of SU collisions on their transmission opportunities. Due to space limitations, we consider only the 0/X- and 1/X-models, which have a high and a low noncooperation cost, respectively. In practice, fair distributed MAC can be realized through random backoff and control message exchange at the cost of some control overhead and/or missed transmission opportunities.

Next, we consider computation of the value function $U^n$ of the coalition $\eta\subseteq S$ on channel $n$, defined as the overall successful transmission probability of $\eta$ on this channel. Given a partition $\rho^n$ on channel $n$, label all other bottom-layer coalitions in $\rho^n \backslash \{ \eta \}$ on channel $n$ as
\begin{equation}\label{Eq:ParitionExcept}
\begin{split}
\rho^n \backslash \{ \eta \}  {=} \{ {\xi _1},{\xi _2}, \ldots ,{\xi _{|\rho^n |- 1}}\}
\end{split}
\end{equation}
and define a binary-valued random vector of length $|\rho^n |{-} 1$ 
\begin{equation}\label{Eq:X}
\begin{split}
\mathbf{X}_{\rho^n\backslash\{\eta\}}=(X_{\xi_1},X_{\xi_2},\ldots,X_{\xi_{|\rho^n|-1}})\in\{0,1\}^{|\rho^n |- 1}
\end{split}
\end{equation}
where $X_{\xi_i}\!\in\!\{0,1\}$ is an indicator variable for the event that the coalition $\xi_i$ experiences a false alarm (FA), i.e.,
\begin{equation}\label{Eq:Indicator}
\Pr(X_{\xi_i}=x_i)=\Bigg\{\begin{matrix}
P_{\mathrm{FA}}^n({\xi _i}), & \textrm{if }x_i=1 \\ 
1-P_{\mathrm{FA}}^n({\xi _i}), & \textrm{if }x_i=0.
\end{matrix}
\end{equation}
The coalition values (successful transmission probabilities) for the two MAC models can be expressed as:
\ifdoublecolumn
\begin{align}\label{Eq:CoalValue0X}
U_\mathrm{0/X}^n(\eta;\rho^n)&=\Pr[\text{$\eta$ is the only coalition in $\rho^n$ that}\notag\\
&\text{successfully detects a spectrum opportunity}]\notag\\
&=\beta^n(1-P_\mathrm{FA}^n(\eta))\cdot{\prod\nolimits _{i=1}^{|\rho^n|-1}}P_\mathrm{FA}^n(\xi_i)
\end{align}
\else
\begin{align}\label{Eq:CoalValue0X}
U_\mathrm{0/X}^n(\eta;\rho^n)&=\Pr[\text{$\eta$ is the only coalition in $\rho^n$ that successfully detects a spectrum opportunity}]\notag\\
&=\beta^n(1-P_\mathrm{FA}^n(\eta))\cdot{\prod\nolimits _{i=1}^{|\rho^n|-1}}P_\mathrm{FA}^n(\xi_i)
\end{align}
\fi
and
\ifdoublecolumn
\begin{align}\label{Eq:CoalValue1X}
&U_\mathrm{1/X}^n(\eta ;\rho^n )=\Pr[\text{$\eta$ successfully detects a spectrum}\notag\\
&\text{ opportunity and wins the MAC contention within $\rho^n$}]\notag\\
&=\beta ^n(1-P_\mathrm{FA}^n(\eta ))\cdot\mathbb{E}\bigg[\frac{|\eta|}{|\eta|+J_{\rho^n\backslash\{\eta\}}\left(\mathbf{X}_{\rho^n\backslash\{\eta\}}\right)}\bigg]\\
&=\!\beta ^n(1\!-\!P_\mathrm{FA}^n(\eta ))\cdot\smashoperator[l]{\sum_{\mathbf{x}{\in}\{0,1\}^{|\rho^n|-1}}}\bigg\{\frac{|\eta|\prod \nolimits_{i=1}^{|\rho^n|-1}\Pr(X_{\xi_i}{=}x_i)}{|\eta|+J_{\rho^n\backslash\{\eta\}}(\mathbf{x})}\bigg\}\notag
\end{align}
\else
\begin{align}\label{Eq:CoalValue1X}
U_\mathrm{1/X}^n(\eta ;\rho^n )& =\Pr[\text{$\eta$ successfully detects a spectrum opportunity and wins the MAC contention within $\rho^n$}]\notag\\
&=\beta ^n(1-P_\mathrm{FA}^n(\eta ))\cdot\mathbb{E}\bigg[\frac{|\eta|}{|\eta|+J_{\rho^n\backslash\{\eta\}}\left(\mathbf{X}_{\rho^n\backslash\{\eta\}}\right)}\bigg]\notag\\
&=\!\beta ^n(1\!-\!P_\mathrm{FA}^n(\eta ))\cdot\smashoperator[l]{\sum_{\mathbf{x}{\in}\{0,1\}^{|\rho^n|-1}}}\bigg\{\frac{|\eta|\prod \nolimits_{i=1}^{|\rho^n|-1}\Pr(X_{\xi_i}{=}x_i)}{|\eta|+J_{\rho^n\backslash\{\eta\}}(\mathbf{x})}\bigg\}
\end{align}
\fi
where $\beta^n$ and $P_\mathrm{FA}^n(\eta)$ are defined in Table~\ref{Table:Notation}, and the number of competing SUs for cognitive access is given by
\begin{equation}\label{Eq:J}
\begin{split}
J_{\rho^n\backslash\{\eta\}}(\mathbf{x})={\sum\nolimits_{i = 1}^{|\rho^n | - 1}} ( {1 - x_i)|{\xi _i}|}. 
\end{split}
\end{equation}

Since the bottom-layer coalition value $U^n(\eta;\rho^n)$ depends not only on the actions of the coalition members, but also on the actions of other SUs on channel $n$, this coalitional game is in partition form \cite{CoalGameSurvey}. We assume the overall values (\ref{Eq:CoalValue0X}) and (\ref{Eq:CoalValue1X}) of a bottom-layer coalition are dividable and can be transferred among coalition members according to the allocated payoffs $a^{mC^n}$ [cf. (\ref{Eq:PTU})], so that the bottom-layer game has TU. To illustrate, consider the following example. Suppose the overall value of a two-SU bottom-layer coalition $\eta=\{1,2\}$ on channel $n$ is $U^n(\eta)=0.8$, and both SUs agree on a payoff allocation of $a^{1C^n}=0.2$ and $a^{2C^n}=0.6$. Under the 0/X-model, the coalition $\eta$ acquires a transmission opportunity when $\eta$ discovers an idle time slot. On the other hand, under the 1/X-model, the coalition $\eta$ acquires a transmission opportunity when $\eta$ discovers an idle time slot, possibly competes with other coalitions, and finally wins a right to transmit. The acquired transmission opportunities are then shared between the two SUs according to their payoff allocation agreement, e.g., in a probabilistic manner, with the conditional probabilities that SU $1$ and SU $2$ are allowed to transmit for the entire time slot given by $a^{1C^n}{/}(a^{1C^n}{+}a^{2C^n}){=}1/4$ and $a^{2C^n}{/}(a^{1C^n}{+}a^{2C^n}){=}3/4$, respectively. Other schedule-based multi-access methods for sharing the transmission opportunity can be employed. For example, SU $1$ and SU $2$ can transmit in a time-division multiple access (TDMA) manner, for $1/4$ and $3/4$ of the time, respectively\footnote{SU transmissions might be unsuccessful either due to SU collisions (under the 0/X-model) or MD of PU traffic (under both MAC models).}. 

\subsection{Cooperative Sensing}\label{subsec:CoopSenseDef}
Next, we present the cooperative sensing scheme at the physical layer for the set of SUs $S$ sensing channel $n$ assuming a bottom-layer partition $\rho^n$ of $S$. The PU transmission is interrupted if one or more bottom-layer coalitions in $\rho^n$ fails to detect its presence, so the integrated MD probability on channel $n$ is given by
\begin{equation}\label{Eq:ChMD}
\begin{split}
P^n_\mathrm{MD}(\rho^n)=1-\prod_{\eta\in\rho^n}\left ( 1-P^n_\mathrm{MD}(\eta) \right ).
\end{split}
\end{equation}
We impose the following constraints:

\noindent \emph{Network-level} constraint \textbf{(C.1)}: $P_\mathrm{MD}^n(\rho^n)=P_\mathrm{MD}^\mathrm{Ch}$ (cf. Table~\ref{Table:Notation})

\noindent \emph{Coalition-level} constraint \textbf{(C.2)}: All equal-sized bottom-layer coalitions should maintain the same MD rate.

\noindent \emph{Node-level} constraint \textbf{(C.3)}: All SUs within a bottom-layer coalition must satisfy the same MD constraint\cite{RLiuEvoGame,YCLiangWCNC}.

The constraints \textbf{(C.1)} and \textbf{(C.2)} are due to the requirement that the set of SUs $S$ sensing channel $n$ should always maintain the integrated MD rate $P_\mathrm{MD}^\mathrm{Ch}$, regardless of how the coalition formation process evolves, which can be satisfied by
\begin{equation}\label{Eq:EtaMD}
P_\mathrm{MD}^n({\eta})=1-(1-P_\mathrm{MD}^\mathrm{Ch})^{|\eta|/{|S|}}
\end{equation}
for all $\eta\in\rho^n$. The constraint \textbf{(C.3)} is motivated by fairness since all SUs should take equal responsibility for PU protection.

Next, we evaluate performance of the \emph{AND-} and the \emph{OR-combining} rules for the proposed constrained spectrum sensing approach. The coalition-level FA probabilities of these fusion rules can be obtained from \cite[eq.~(10--17)]{YCLiangWCNC} assuming AWGN channels\footnote{In practice, these channels may be subject to fading, and the ergodic sensing accuracy probabilities can be obtained by averaging over the fading distribution \cite{CoopSense,MultiChCoalSenseGame,ZHanCoalSenseGame}.}:
\begin{align}
P_\mathrm{FA,AND}^n(\eta ) &=\smashoperator{\prod\limits_{m \in \eta }}P_\mathrm{FA,AND}^n(m)\label{Eq:CoalFAResultAND}\\
&= \smashoperator{\prod\limits_{m \in \eta }}Q\left(\sqrt{2\lambda^{mn}+1}\phi+\sqrt\nu\lambda^{mn}\right)\notag\\
P_\mathrm{FA,OR}^n(\eta ) &=1 - \smashoperator{\prod\limits_{m \in \eta }}(1-P_\mathrm{FA,OR}^n(m))\label{Eq:CoalFAResultOR}\\
&= 1 - \smashoperator{\prod\limits_{m \in \eta }}{\left(1-Q(\sqrt{2\lambda^{mn}+1}\tilde{\phi}+\sqrt\nu\lambda^{mn})\right)}\notag
\end{align}
respectively, with [cf. (\ref{Eq:EtaMD}) and \textbf{(C.3)}]
\begin{align}
\phi  &\triangleq Q^{-1}\big(1-P_\mathrm{MD,AND}^n(m)\big)=Q^{-1}\big( {( 1 - P_\mathrm{MD}^n(\eta ) )}^\frac{1}{|\eta |} \big)\notag\\
&=Q^{-1}\big((1-P_\mathrm{MD}^\mathrm{Ch})^\frac{1}{|S |}\big)\label{Eq:BetaDef}\\
\tilde \phi & \triangleq Q^{-1}\big(1-P_\mathrm{MD,OR}^n(m)\big)= Q^{- 1}\big( 1 - ( {P_{\mathrm{MD}}^n(\eta )} )^\frac{1}{|\eta|}\big)\notag\\
&= Q^{- 1}\Big( 1 - \big( 1-(1-P_\mathrm{MD}^\mathrm{Ch})^{|\eta|/{|S|}} \big)^\frac{1}{|\eta|}\Big)\label{Eq:TildeBetaDef}
\end{align}
where $m\in \eta\in \rho^n$, and $Q(\cdot)$ is the Q-function \cite[eq. (B.20)]{GoldsmithWireless}. First, from (\ref{Eq:BetaDef}) and (\ref{Eq:TildeBetaDef}), the individual MD constraint $P_\mathrm{MD}^n(m)$ depends on the size of the bottom layer coalition to which an SU belongs for the OR-rule, but not for the AND-rule. Thus, when using the AND-rule, SUs do not have to update the MD constraints when computing their values within different bottom-layer coalitions as long as the SU population $|S|$ on channel $n$ does not change. Second, for both rules, as $|S|$ increases, each coalition $\eta$ reduces its MD rate $P_\mathrm{MD}^n(\eta)$ in (\ref{Eq:EtaMD}), leading to increased FA probability $P^n_\mathrm{FA}(\eta)$ in (\ref{Eq:CoalFAResultAND},\ref{Eq:CoalFAResultOR}) and decreased coalition values in (\ref{Eq:CoalValue0X},\ref{Eq:CoalValue1X}). Thus, large values of $|S|$ are penalized, balancing SU competition on all channels. 

Finally, our extensive simulation results using (\ref{Eq:CoalFAResultAND}--\ref{Eq:TildeBetaDef}) show that the AND-rule is more suitable when cooperative SUs have heterogeneous sensing capabilities (i.e., different PU-to-SU SNRs) while the OR-rule provides better performance in homogeneous scenarios. Intuitively, if there is at least one SU $m\in\eta$ with favorable PU-to-SU SNR $\lambda^{mn}$, and thus with small FA probability $P^n_\mathrm{FA,AND}(m)$, the resulting $P^n_\mathrm{FA,AND}(\eta)$ is very small since $P^n_\mathrm{FA,AND}(\eta)\leq \min_{m\in \eta}P^n_\mathrm{FA,AND}(m)$ in (\ref{Eq:CoalFAResultAND}). In contrast, for the OR-rule, if only one member SU has very poor sensing capacity, the entire bottom-layer coalition $\eta$ suffers since $P^n_\mathrm{FA,OR}(\eta)\geq \max_{m\in \eta}P^n_\mathrm{FA,OR}(m)$ in (\ref{Eq:CoalFAResultOR}). For example, in Fig.~\ref{Fig:Hetero} we compare the integrated FA probabilities (\ref{Eq:CoalFAResultAND}) and (\ref{Eq:CoalFAResultOR}) for a coalition $\eta=\{1,2\}$ on channel $n$ when the average PU-to-SU SNR of the two member SUs is fixed to $\bar\lambda$. We observe that the OR-rule and the AND-rule achieve their best performance for homogeneous and heterogeneous sensing capacities (given by PU signal strengths at the sensor) of the two SUs, respectively.  The general case is explored below\cite{YLuThesis}.

\begin{figure}[!t]
    \centering
    \includegraphics[scale=\FigScale]{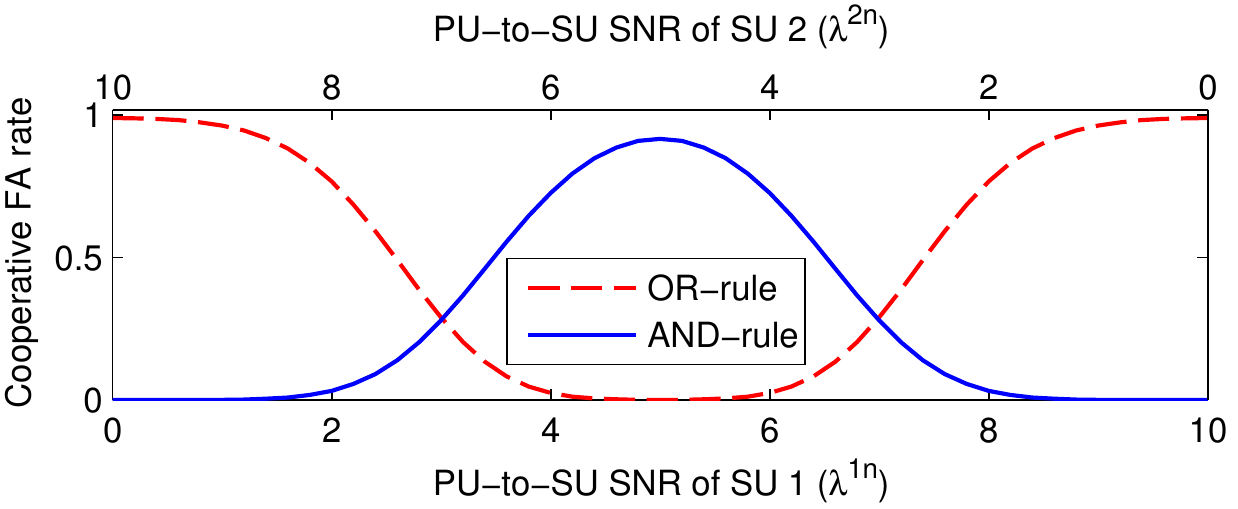}
    \caption{Performance comparison of the $P_{\mathrm{FA}}^n(\eta )$ for the AND- and OR-rule under the proposed MD constraints; $\eta=\{1,2\}$; $P_{{\mathrm{MD}}}^n(\eta )=10^{-4}$; $\nu=5$; $\bar{\lambda}=0.5(\lambda^{1n}+\lambda^{2n})=5$ ($7$\,dB).}
    \label{Fig:Hetero}
\end{figure}

\begin{Prop} \label{Prop:Convex} Consider a coalition $\eta$ on channel $n$ and a fixed bottom-layer-coalition-level MD probability constraint $P_{{\mathrm{MD}}}^n(\eta )$\footnote{Due to space limitations, we omit the proof, which shows the desired quasiconcavity and quasiconvexity by establishing the logconcavity of $P_\textnormal{FA,AND}^n(\eta )$ and $1-P_\textnormal{FA,OR}^n(\eta )$, respectively\cite{YLuThesis}.}. The FA probability of $\eta$ for the AND-rule (\ref{Eq:CoalFAResultAND}) [OR-rule (\ref{Eq:CoalFAResultOR})] is a quasiconcave (quasiconvex) function \cite[Section 3.4]{ConvexOpt}, which achieves its global maximum (global minimum) when $\lambda^{mn}=\bar{\lambda}$, $\forall m\in\eta$, subject to an average PU-to-SU SNR constraint
\begin{equation}\label{Eq:SNRConstraint}
\frac{1}{{|\eta |}}\sum\limits_{m \in \eta } {{\lambda ^{mn}}}  = \bar \lambda  
\end{equation}
if $\forall m\in\eta$,
\begin{equation}\label{Eq:ConvexCondition1}
P_{{\mathrm{MD}}}^n(\eta )< {(0.5)^{|\eta |}} 
\end{equation}
\begin{equation}\label{Eq:ConvexCondition2}
\tilde{\phi}>-\frac{\sqrt{\nu}(2\lambda^{mn}+1)^\frac{3}{2}}{3\lambda^{mn}+2}.
\end{equation}
\end{Prop}
Note that the constraints (\ref{Eq:ConvexCondition1}) and (\ref{Eq:ConvexCondition2}) are mild\cite{YLuThesis}. For example, they can be easily satisfied if: (i) $P^n_\mathrm{MD}(\eta)\geq 10^{-4}$ and $|\eta|\leq 13$ and (ii) ($\nu\geq 56$ for $\lambda^{mn}\geq 0$) or ($\nu\geq 13$ for $\lambda^{mn}\geq 1$) or ($\nu\geq 2$ for $\lambda^{mn}\geq 10$), i.e., when the MD constraint  $P^n_\mathrm{MD}(\eta)$ is not extremely stringent, the coalition size $|\eta|$ is not too large, and the number of collected samples for spectrum sensing $\nu$ is sufficiently large{\footnote{In practice, low PU-to-SU SNR (e.g. $\leq 1$) requires $\nu$ on the order of hundreds to thousands and medium-to-high PU-to-SU SNR (e.g. $\geq 10$) requires less than $10$ samples for good performance (See, e.g., \cite{Tradeoff,CoopSense}).}}.

Based on these observations, in the remainder of this paper we will assume the AND fusion rule to study a typical wireless CR network with a moderate number of cooperating SUs on each channel, which experience heterogeneous PU-to-SU channel conditions.

\section[Two-Layer Coalitional Sensing and Access Game]{Two-Layer Coalitional Sensing and Access Game\footnote{Proofs of Propositions~\ref{Prop:Characteristic} and \ref{Prop:GrandCoalFormation}--\ref{Prop:Converge} can be found in\cite{YLuCISS}.}}\label{Sec:Game}

\subsection{Grand Coalition Formation and Payoff Allocation in the Bottom-layer Game}\label{SubSec:BottomSol}
Consider the bottom-layer game $(S,U^n)$ on channel $n$.

\begin{Prop}\label{Prop:Characteristic}
For the 0/X-model, the game $(S,U_\mathrm{0/X}^n)$: 
\begin{enumerate}[(i)]
	\item reduces to a \emph{characteristic-form game}, i.e., the value of a bottom-layer coalition $\eta$ depends solely on the composition of $\eta$ \cite{CoalGameSurvey}: $U_\mathrm{0/X}^n(\eta;\rho^n)=U_\mathrm{0/X}^n(\eta;\tilde\rho^n)\triangleq U_\mathrm{0/X}^n(\eta)$, for any $\eta \subseteq S$ and any bottom-layer partitions $\rho^n$ and $\tilde\rho^n$ of $S$, such that $\eta \in \rho^n$  and  $\eta \in \tilde\rho^n$; 

\item is \emph{superadditive} \cite{CoalGameSurvey}, i.e., the SUs benefit from forming larger coalitions. Thus, for any two disjoint bottom-layer coalitions $\eta,\xi \subset S$, $U_\mathrm{0/X}^n(\eta \cup \xi)\geq U_\mathrm{0/X}^n(\eta)+U_\mathrm{0/X}^n(\xi)$; 

\item has efficient grand coalition \cite{Efficiency}, i.e., the grand coalition value is at least as large as the combined value of all coalitions in any other partition: $\forall\tilde\rho^n\neq\{S\}$, $U_\mathrm{0/X}(S)\geq \sum_{\eta\in\tilde\rho^n} U_\mathrm{0/X}(\eta)$. This property follows from (ii).
\end{enumerate}
\end{Prop}

\begin{Prop} \label{Prop:EqualEfficient}
Under the 1/X-model, all bottom-layer partitions of $S$ are equally efficient, i.e., for any two partitions $\rho^n$ and $\tilde\rho^n$ of $S$, $\sum\nolimits_{\eta  \in \rho^n } {U_\mathrm{1/X}^n(\eta ;\rho^n)}  = \sum\nolimits_{\eta  \in \tilde\rho^n } {U_\mathrm{1/X}^n(\eta ;\tilde\rho^n )}$.
\end{Prop}
\begin{proof}
Under the 1/X-model, a spectrum opportunity on channel $n$ is wasted if all bottom-layer coalitions make FAs on the PU presence. Thus, for any bottom-layer partition $\rho^n$ of any $S$ on any channel $n$ [cf. (\ref{Eq:CoalFAResultAND})]
\begin{align}
\sum\limits_{\eta  \in \rho^n } {U_\mathrm{1/X}^n(\eta ;\rho^n )}  =& \beta^n(1 - \prod\limits_{\eta  \in \rho^n } {P_{{\mathrm{FA}}}^n(\eta )}\notag\\
=& \beta^n(1 - \prod\limits_{m \in S} {P_{{\mathrm{FA}}}^n(m)}) \label{Eq:EqualEfficientProof}
\end{align}
which is independent of $\rho^n$.
\end{proof}

\begin{Prop} \label{Prop:NegExternality}
Under the 1/X-model, the bottom-layer game $(S,U_\mathrm{1/X}^n)$ exhibits \emph{nonpositive externalities}\cite{Maskin}, i.e., a merger between two coalitions cannot improve the values of other coalitions. Formally, for any disjoint coalitions $\eta ,\xi ,\zeta  \subset S$ and any partition $\rho^n$ of $S$ such that  $\eta ,\xi ,\zeta  \in  \rho^n$, the coalition value
$U_\mathrm{1/X}^n(\eta ;\rho^n )  \geq U_\mathrm{1/X}^n\big(\eta ; \rho^n \backslash \{ \xi\} \backslash \{\zeta \}\cup \{ \xi  \cup \zeta \} \big)$.
\end{Prop}
\begin{proof}
See Appendix~\ref{Sec:NegExternalityProof}.
\end{proof}

\begin{Prop}\label{Prop:GrandCoalFormation}
The grand coalition always forms in the bottom-layer game under the 1/X and 0/X MAC models.
\end{Prop}
The proof of Proposition~\ref{Prop:GrandCoalFormation} is based on Propositions~\ref{Prop:Characteristic}--\ref{Prop:NegExternality} and follows from the superadditivity property \cite{CoalGameSurvey} for the 0/X model and the \emph{weak efficiency} or the nonpositive externality property \cite{Maskin, Efficiency} for the 1/X model \cite{YLuCISS}\footnote{There are several different points of view on the grand coalition formation in the game theory literature. Here we employ the approaches in\cite{CoalGameSurvey,Maskin, Efficiency}.}. Proposition~\ref{Prop:GrandCoalFormation} implies that all SUs sensing channel $n$ cooperate, and MAC is not utilized. However, MAC needs to be assumed to compute the hypothetical payoff that an SU \emph{could have obtained} if it left the grand coalition and/or the \emph{marginal value} it brings into the grand coalition. Such information is crucial when determining the actual payoff allocation to ensure the stability of and the fairness within the grand coalition \cite{GameSurvey}.

Next, we discuss the allocation of the value of the grand coalition $U^n(S;\{S\})$ (hereafter referred to as $U^n(S)$ for brevity) while satisfying every member SU with its allocated payoff. Fair payoff allocation rules for traditional single-layer coalitional games are extensively studied in the literature assuming (hypothetical) breakdown of the grand coalition, e.g., the Shapley value \cite{CoalGameSurvey}, the Owen value \cite{Owen}, the nucleolus\cite{CoalGameSurvey}, the Nash bargaining solution (NBS)\cite{NBS,NBSNeg}, etc. In this paper, we employ NBS due to its computational efficiency\cite{NBS}. The NBS payoff allocation method distributes the overall value of the grand coalition among the individual players according to the \emph{disagreement point}, which characterizes the minimum payoffs the players are willing to accept. Given the disagreement point, NBS splits the overall profit of cooperation equally among the players and maximizes the product of the players' profits\cite{NBS}.

\begin{Prop}\label{Prop:NBS}	For the characteristic-form (cf. Proposition \ref{Prop:Characteristic}) bottom-layer game $(S,U_\mathrm{0/X}^n)$, NBS assigns a payoff, $\forall m\in S$ \cite{NBS,NBSNeg}
\begin{equation}\label{Eq:NBS}
\begin{split}
a^{mC}_\mathrm{NBS,0/X}=\frac{U_\mathrm{0/X}^n(S)-\sum\limits_{i\in S}U_\mathrm{0/X}^n(\{i\})}{|S|}+ U_\mathrm{0/X}^n(\{m\})
\end{split}
\end{equation}
where the disagreement point is $\big\{U_\mathrm{0/X}^n(\{m\}):m\in S \big\}$.
\end{Prop}
\noindent The $m$th component of the disagreement point $U_\mathrm{0/X}^n(\{m\})$ is chosen as the payoff SU $m$ can obtain with noncooperative sensing, i.e., if bargaining fails. From Proposition~\ref{Prop:Characteristic}, $U_\mathrm{0/X}^n(\{m\})\triangleq U_\mathrm{0/X}^n(\{m\};\rho^n)$ for any partition $\rho^n$ that includes $m$ as a singleton, i.e., it does not depend on actions of other SUs. Since SUs with stronger sensing capabilities have higher singleton values $U_\mathrm{0/X}^n(\{m\})$, the NBS (\ref{Eq:NBS}) provides each SU with the allocated payoff commensurate with its sensing contribution, resulting in a fair payoff allocation.

When the 1/X-model is assumed, the selfish individual payoff $U_\mathrm{1/X}^n(\{m\};\rho^n)$ depends not only on that SU's sensing capability, but also on other SUs' reactions when a given SU does not cooperate. In this paper, we employ the \textit{fine} NBS (\textit{f}NBS) \cite{NBSNeg}, which implies all SUs stand alone if bargaining is not successful, i.e., if the grand coalition breaks down.

\begin{Prop} \label{Prop:fNBS}Under the 1/X model, the payoff of each SU $m\in S$ using \textit{f}NBS is \cite{NBSNeg}
\begin{align}\label{Eq:fNBS}
a_{{\mathrm{\textit{f}NBS,1/X}}}^{mC}&=\frac{1 }{|S|}\bigg[U_\mathrm{1/X}^n(\{ S\} ) - \smashoperator{\sum\limits_{i \in S}} U_\mathrm{1/X}^n\left(\{ i\} ;\{ \{ j\} :{j {\in} S}\} \right)\bigg]\notag\\
&\quad\quad\quad+U_\mathrm{1/X}^n (\{ m\} ;\{ \{ j\} :{j {\in} S}\})\notag\\
&= U_\mathrm{1/X}^n(\{ m\} ; \{ \{ j\} :{j {\in} S}\} ).
\end{align}
\end{Prop}
Thus, the allocated payoff probability is simply given by the successful transmission probability that each SU could have obtained individually assuming all other SUs have also formed singletons. However, this does not imply that SUs should deviate from the grand coalition, because each SU is not guaranteed a payoff of $ a_{{\mathrm{\textit{f}NBS,1/X}}}^{mC}$ if it deviates. In fact, if an SU decides to remain isolated, its payoff is at risk due to the nonpositive externalities (Proposition~\ref{Prop:NegExternality}) if other SUs collude. As a result, an SU $m$ might end up with a much worse payoff than its singleton value $U_\mathrm{1/X}^n(\{ m\} {;}{ \{ \{ j\} :{j {\in} S}\}})$. Therefore, every SU has an incentive to join the grand coalition. Moreover, we found that a \textit{coarse} NBS (\textit{c}NBS) \cite{NBSNeg} allocation where all other SUs form a coalition and try to reduce the hypothetical individual payoff of the deviating SU, results in similar throughput and energy consumption as that of the \textit{f}NBS method in Proposition~\ref{Prop:fNBS} \cite{YLuThesis}.

\renewcommand{\algorithmicrequire}{\textbf{Input:}}
\renewcommand{\algorithmicensure}{\textbf{Output:}}
\begin{figure}[!t]

\hrule\hrule\vspace{\Extra}
\textbf{Algorithm 1} Bottom-layer payoff allocation
\vspace{\Extra}
\hrule
\begin{algorithmic}[1]
\Require $C{=}(S,n)$, $\{\lambda^{mn}:m{\in}S\}$, $P_\mathrm{MD}^\mathrm{Ch}$, $\nu$, $\beta^n$
\Ensure $\{a^{mC}:m{\in}S\}$
\State $P^n_\mathrm{FA}(m){=}Q\Big(\sqrt{2\lambda^{mn}{+}1}Q^{-1}((1-P_\mathrm{MD}^\mathrm{Ch})^\frac{1}{\left | S \right |}){+}\lambda^{mn}\sqrt{\nu}\Big)$\label{Line:AdaptTh}
\If   {0/X-MAC}{ $\forall m\!\in\! S$:}
\State $U_\mathrm{0/X}^n(S)=\beta^n(1-\prod_{i{\in}S}P_\mathrm{FA}^n(i))$
\State $U_\mathrm{0/X}^n(\{m\}){=}\beta^n(1-P_\mathrm{FA}^n(m))\prod _{i\in S\backslash\{m\}}P_\mathrm{FA}^n(i)$ 
\State Compute $a^{mC}_\mathrm{NBS,0/X}$ using  (\ref{Eq:NBS})
\EndIf
\If {1/X-MAC}{ $\forall m\!\in\! S$:}
\State {$a_{{\mathrm{\textit{f}NBS,1/X}}}^{mC} = U_\mathrm{1/X}^n(\{ m\} ;\{ \{ j\} :{j {\in} S}\} )$ as in (\ref{Eq:CoalValue1X})}
\EndIf
\end{algorithmic}
\hrule
\caption{Bottom-layer payoff allocation.}
\label{Fig:AlgBottom}
\end{figure}

Moreover, the bargaining is always successful, i.e., $a^{mC^n}\geq U^n(\{ m\})$ for all $m$ and both MAC models (cf. Propositions~\ref{Prop:Characteristic}--\ref{Prop:fNBS} and \cite{NBSNeg}). Once the negotiation is settled, all SUs in $S$ sign an agreement on the sharing of true sensing results as well as the payoff allocation $\{a^{mC^n}:m\in S\}$. Also specified in this agreement are the transmission coordination protocol and the $m$th SU's fair share of a sensed available slot, given by $\{a^{mC^n}/{\sum_{i\in S}}a^{iC^n}:m\in S\}$ (see Example at the end of Section~\ref{SubSec:BottomDef}). All coalition members should thereafter strictly conform to the finalized agreement\footnote{We assume the SUs are selfish, but honest. In practice, an SU may falsify its sensing result or otherwise breach an earlier agreement, which can be detected  by, e.g., a reputation-based approach \cite{Reputation}.}.  The bottom-layer payoff allocation algorithm is summarized in Fig.~\ref{Fig:AlgBottom}.

\subsection{Coalition Formation at the Top Layer}\label{SubSec:TopSol}
In the top-layer game, SUs use (\ref{Eq:Preference}) to determine if switching to another channel is advantageous. The proposed algorithm in Fig.~\ref{Fig:AlgTop} employs a distributed switching scheduling scheme to facilitate fast convergence to an SU network partition. At most one switch is allowed in each time slot. Initially, all SUs actively compete for the right to switch using an out-of-band control channel \cite{ContrChSurvey}. In every time slot, the winning SU $m$ randomly chooses a potential new channel and notifies other SUs in $\mathcal{M}\backslash \{m\}$ about its action by broadcasting (on the control channel) (i) a SWITCH signal if (\ref{Eq:Preference}) holds, (ii) a HOLD signal if (\ref{Eq:Preference}) fails and it still plans to continue searching, or (iii) a SLEEP signal when all switching opportunities have been exhausted unsuccessfully, i.e., the current channel is the most preferable for this SU. Competition for the right to switch continues until all SUs are asleep, indicating convergence of the partition formation process. Note that all SUs (including sleeping SUs) should cognitively monitor the environment changes, e.g., variation of the PU-to-SU SNRs, the PU/SU locations, the number of SUs, the number of channels, etc., and repeat the partition formation algorithm in Fig.~\ref{Fig:AlgTop} when changes occur (See Section~\ref{Sec:Sim}). Assuming that these parameters are fixed, the following Proposition holds:
\begin{Prop}\label{Prop:Converge} The proposed coalition formation algorithm (Fig.~\ref{Fig:AlgTop}) converges in at most $N^M$ switches.
\end{Prop}

\renewcommand{\algorithmicrequire}{\textbf{Input:}}
\renewcommand{\algorithmicensure}{\textbf{Output:}}
\begin{figure}[!t]
\hrule\hrule\vspace{\Extra}
\textbf{Algorithm 2} Distributed top-layer partition formation
\vspace{\Extra}
\hrule
\begin{algorithmic}[1]
\Require $\mathcal{M},\mathcal{N},\beta^\mathcal{N},\gamma^{\mathcal{M}\times \mathcal{N}},\lambda^{\mathcal{M}\times \mathcal{N}}$
\Ensure $\mathcal{P}$
\State \textbf{Initialization}: Each SU $m$ randomly senses a channel $n_*^m$
\State $C^n(1)=\{m:n_*^m=n\}$ and $C^n(2)=n$, $\forall n\in \mathcal{N}$
\State $\mathcal{P}=\{C^1,\cdots, C^N\}$ and $\textrm{Action}=\textrm{SWITCH}$
\While{$\mathcal{M}_\textrm{Active}\neq\emptyset$ \textbf{(at each time slot)}:}\label{Line:While}
\If {$\textrm{Action}=\textrm{SWITCH}$}: {$\mathcal{M}_\mathrm{Active}=\mathcal{M}$}
\EndIf
\If {$\textrm{Action}\neq\textrm{HOLD}$}:
\State SUs in $\mathcal{M}_\mathrm{Active}$ contend for the right to switch
\State SU $m\in C^n(1)$ wins and $\mathcal{N}_\mathrm{Candidate}=\mathcal{N}\backslash \{n\}$
\EndIf
\State $m$ randomly chooses channel ${\tilde{n}}\in\mathcal{N}_\mathrm{Candidate}$
\State $\tilde{C}^n=(C^n(1)\backslash \{m\},n)$, $\tilde{C}^{\tilde{n}}=(C^{\tilde{n}}(1){\cup}\{m\},\tilde{n})$
\State $m$ plays the bottom-layer game $(\tilde{C}^{\tilde{n}},U^{\tilde{n}})$ in Fig.~\ref{Fig:AlgBottom}\label{Line:Play}
\If   {(\ref{Eq:Preference}) holds}:\label{Line:CheckPrefer}
\State $a^{iC^n}=a^{i\tilde{C}^n}$, $a^{jC^{\tilde{n}}}=a^{j\tilde{C}^{\tilde{n}}}$,  $\forall i\in C^n(1)$, $j\in C^{\tilde{n}}(1)$
\State  $\mathcal{P}=\mathcal{P} \backslash\{C^n\}\backslash\{C^{\tilde{n}}\}\cup\{{ \tilde{C}^n}\}\cup \{ \tilde{C}^{\tilde{n}}\}$\State $C^n=\tilde{C}^n$, $C^{\tilde{n}}= \tilde{C}^{\tilde{n}}$,  $\textrm{Action}=\textrm{SWITCH}$
\Else :
\State $m$ stays on channel $n$ and $\mathcal{N}_\mathrm{Candidate}=\mathcal{N}_\mathrm{Candidate}\backslash \{\tilde{n}\}$
\If{$\mathcal{N}_\mathrm{Candidate}\neq\emptyset$}: {$\textrm{Action}=\textrm{HOLD}$}
\Else : {$\mathcal{M}_\mathrm{Active}=\mathcal{M}\backslash \{m\}$, $\textrm{Action}=\textrm{SLEEP}$}
\EndIf
\EndIf
\If{a slot is sensed idle by $C^n=(S,n)$} $\forall m\in S$,
\State{fair share of the slot for $m$th SU's transmission is $a^{mC^n}/{\sum_{i\in S}}a^{iC^n}$}
\EndIf
\EndWhile\label{Line:Converge}
\end{algorithmic}
\hrule
\caption{Distributed top-layer partition formation.}
\label{Fig:AlgTop}
\end{figure}

\subsection{Computational complexity and overhead of the two-layer game}\label{SubSec:Complexity}
The computational complexity of the proposed game is dominated by the FA rate computations (line \ref{Line:AdaptTh} in Fig~\ref{Fig:AlgBottom}). For SU $m$ on channel $n$ [cf. (\ref{Eq:CoalFAResultAND}) and (\ref{Eq:BetaDef})], this rate is determined by $\lambda^{mn}$ and SU population size $|C_t^n(1)|$, where $t$ is the time slot index. Each SU $m$ on channel $n$ maintains three values of its FA rate, corresponding to the current and potential SU populations sizes $|C_t^n(1)|\pm i$, $i=0,1$, and this FA rate information is exchanged over the control channel as needed. Thus, the initialization step requires $3M$ FA rate computations. When an SU explores whether to switch to channel $\tilde{n}$ in time slot $t$ using (\ref{Eq:Preference}), it needs to compute its potential FA rate on that channel, corresponding to the updated SU population size $|C_t^{\tilde{n}}(1)|+1$ and its PU-to-SU SNR $\lambda^{m\tilde{n}}$, resulting in one FA computation per time slot. It then requests the FA rate information from SUs residing on that channel. In those slots where an SU $m$ actually switches to channel $\tilde{n}$, new coalitions $\tilde{C}_t^n$ and $\tilde{C}_t^{\tilde{n}}$ form [cf. (\ref{Eq:Preference})], and SU $m$ needs to compute two additional FA rates corresponding to the SU population sizes $|\tilde{C}_t^{\tilde{n}}(1)|\pm 1$ while all other SUs in $\tilde{C}_t^n(1)$ and $\tilde{C}_t^{\tilde{n}}(1)$ have to compute only one additional FA rate each, corresponding to SU population sizes decreased and increased by 2 on their respective channels $n$ and $\tilde{n}$. Since each FA rate computation takes $\mathcal{O}(1)$ time, the total computational complexity of the algorithm until convergence is
\begin{equation}
\mathcal{O}\Big(3M+T_\textrm{Converge}+\!\sum_{t\in\mathcal{T}_\textrm{Switch}}(|\tilde{C}_t^n(1)|+|\tilde{C}_t^{\tilde{n}}(1)|+1)\Big )
\label{Eq:Complexity}
\end{equation}
where $\mathcal{O}(\cdot)$ is the big O notation\cite{CLRS}, $\mathcal{T}_\textrm{Switch}$ is the set of time slots in which an SU switches in the interval $[1,T_\textrm{Converge}]$ and $T_\textrm{Converge}$ denotes the convergence time, i.e., the number of slots needed to execute the while loop from line~\ref{Line:While} to \ref{Line:Converge} in Fig.~\ref{Fig:AlgTop}.

From the above discussion, the communication overhead of the proposed two-layer game mainly results from the action signaling (SWITCH/HOLD/SLEEP) in Fig.~\ref{Fig:AlgTop}, the FA rate information exchange discussed above, and sharing of sensing results. The latter two types of overhead are commonly assumed in the cooperative sensing literature, e.g. \cite{HedonicSenseGame,AuctionBased,ZHanCoalSenseGame,FairPayoff,MultiChCoalSenseGame,Overlapp,RLiuEvoGame}, while the former costs no more than two bits per slot. Moreover, we must emphasize that the exchange of local data rate information $\{R^{mn}:m\in S\}$ is not required.

We note that the two-layer game structure limits the computational complexity, the negotiation delay, and the communication overhead, since an SU that explores switching from channel $n$ to $\tilde{n}$ needs to learn only the FA rates of the SUs currently residing on channel $\tilde{n}$. In contrast, in a one-layer game where SUs can form arbitrary coalitions spanning multiple channels, an SU looking for a switching opportunity would have to negotiate with and obtain data on a much larger set of SUs.

Finally, while we have ignored the overhead of communication and MAC when computing the coalition values in (\ref{Eq:CoalValue0X}) and (\ref{Eq:CoalValue1X}), such overhead is inevitable in practice. Very large overhead might greatly reduce the value of the grand coalition and, thus, its formation on individual channels. Payoff allocation has not received significant attention in this scenario since most results on bargaining games focus on allocating the value of the grand coalition. Moreover, delay and uncertainty may arise in reaching a joint agreement on both the coalition structure and the payoff allocation in the bottom-layer game. Development of novel game-theoretic solutions is required to accommodate such extensions.

\section{Simulation Results}\label{Sec:Sim}
We assume the following simulation setup throughout this section unless otherwise noted: The sensing and slot durations are $5$\,ms and $100$\,ms, respectively. The SU sensing power $P_S$, SU transmission power $P_\mathrm{SU}$, PU transmission power $P_\mathrm{PU}$, and noise power $P_N$ are $10$\,mW, $10$\,mW, $100$\,mW and $0.1$\,mW, respectively. All users are randomly placed in a square region of $100$\,m${\times}100$\,m, and only path loss effects are considered with the pass loss exponent equal to $2$\cite{HedonicSenseGame}. The PU-to-SU SNR is given by $\lambda^{mn}=P_\mathrm{PU}d_{mn}^{-2}/P_N$, where $d_{mn}$ is the distance between PU $n$ and SU $m$. Similarly, the SU-to-SU SNR between two SUs $m$ and $m'$ is given by $\gamma^{mm'}=P_\mathrm{SU}d_{mm'}^{-2}/P_N$, where $d_{mm'}$ is the SU distance. We compare only with \cite{HedonicSenseGame} because, to the best of our knowledge, other related references did not consider joint spectrum sensing and access in a multichannel CR network.

\begin{figure}[!t]
    \centering
    \includegraphics[scale=\FigScale]{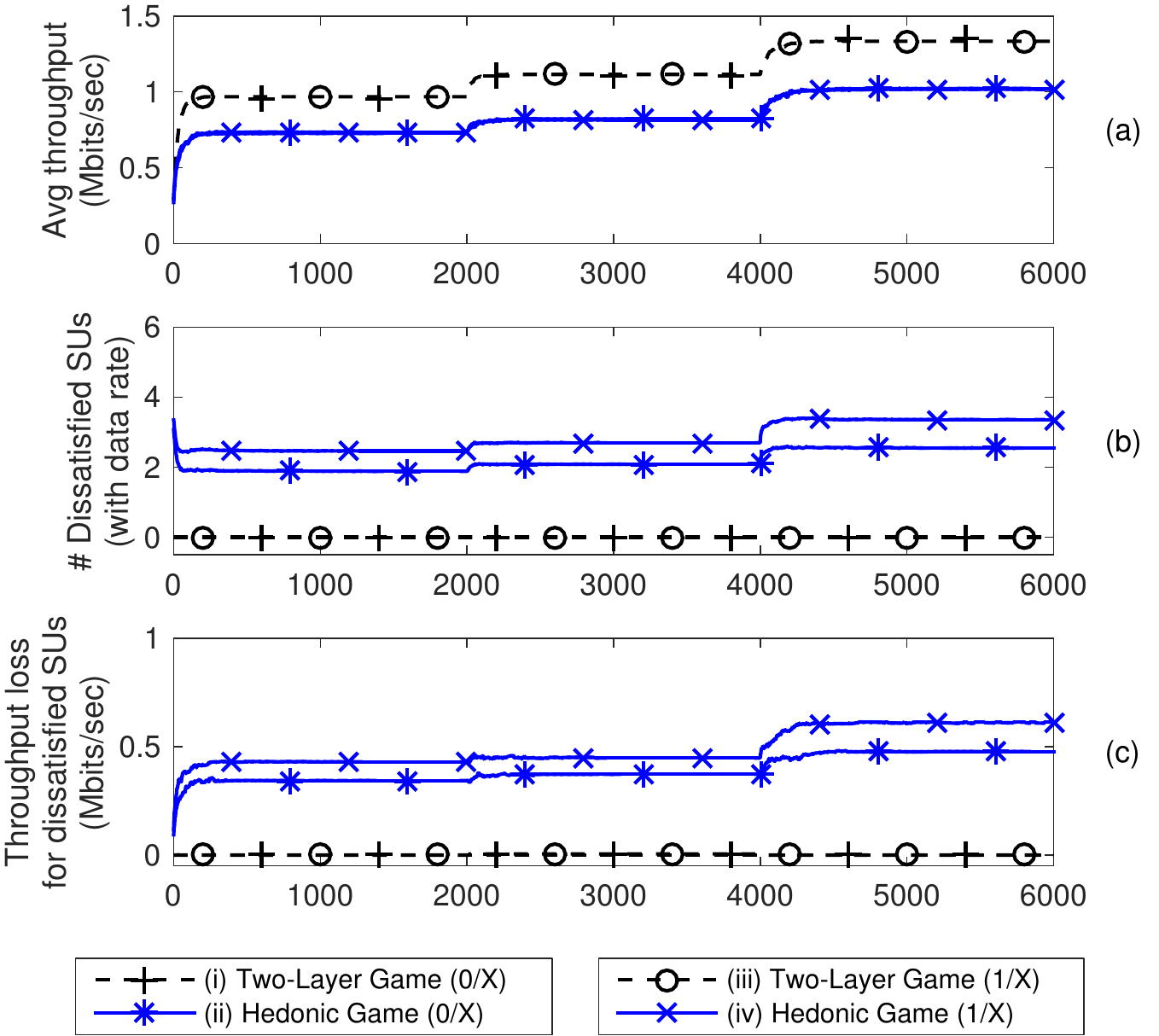}
		\includegraphics[scale=\FigScale]{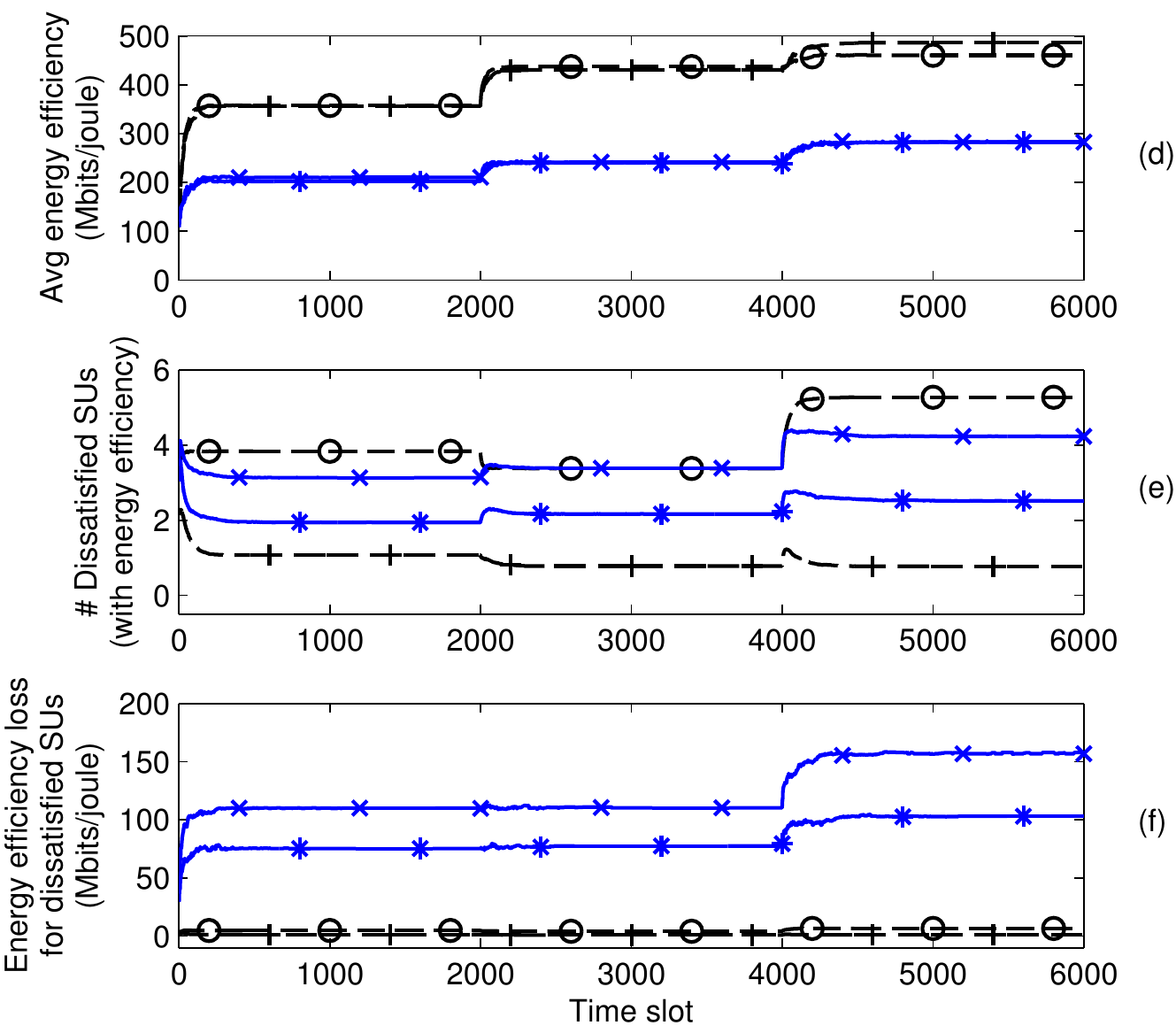}
    \caption{Performance comparison of the proposed two-layer game and the hedonic game in \cite{HedonicSenseGame}.}
    \label{Fig:EnergyEfficiency}
\end{figure}

In Fig.~\ref{Fig:EnergyEfficiency}, we compare the proposed game (i,iii) with the game in \cite{HedonicSenseGame} (ii,iv) for the two MAC models [cf. (\ref{Eq:CoalValue0X},\ref{Eq:CoalValue1X})]. Initially there are $M=10$ SUs and $N=5$ PUs. At time slots $2000$ and $4000$, these parameters change to ($M=10,N=6$) and ($M=14,N=6$), respectively. Each PU uses one channel with $B^n=10$\,MHz exclusively. The channel availability probability is $\beta^n=0.2$. For fair comparison, the AND-rule combining scheme under the constraints described in Section~\ref{subsec:CoopSenseDef} is used for both the proposed two-layer game and the game in \cite{HedonicSenseGame} assuming $\nu=5$ and $P_\mathrm{MD}^\mathrm{Ch}=0.01$. Note that \cite{HedonicSenseGame} aims to maximize the energy efficiency while the objective of the proposed game is to maximize the expected data rate (\ref{Eq:PTU}). Nevertheless, the proposed game achieves better average energy efficiency and throughput in Fig.~\ref{Fig:EnergyEfficiency}(a,d), since
in this game SUs are encouraged to join channels with more detectable PU traffic to obtain higher payoffs, leading to efficient sensing task distribution among the selfish SUs. As a result, more transmission opportunities can be correctly detected in the CR network while maintaining the same energy cost for spectrum sensing.

Moreover, each SU is allocated the access opportunities it deserves and, thus, is provided with sufficient incentives to participate in the proposed game. In particular, as shown in Fig.~\ref{Fig:EnergyEfficiency}(b,c), every SU has higher data rate when playing the proposed game than operating alone (i.e., all SUs are satisfied with their data rates). Despite the fact that the proposed game can have a greater number of dissatisfied SUs in terms of energy efficiency than the game in \cite{HedonicSenseGame} under the 1/X-model [curves iii vs. iv in Fig.~\ref{Fig:EnergyEfficiency}(e)], these SUs experience negligible energy efficiency loss [Fig.~\ref{Fig:EnergyEfficiency}(f)], and the overall energy efficiency is significantly higher for the two-layer game [Fig.~\ref{Fig:EnergyEfficiency}(d)].

\begin{figure}[!t]
    \centering
    \includegraphics[scale=\FigScale]{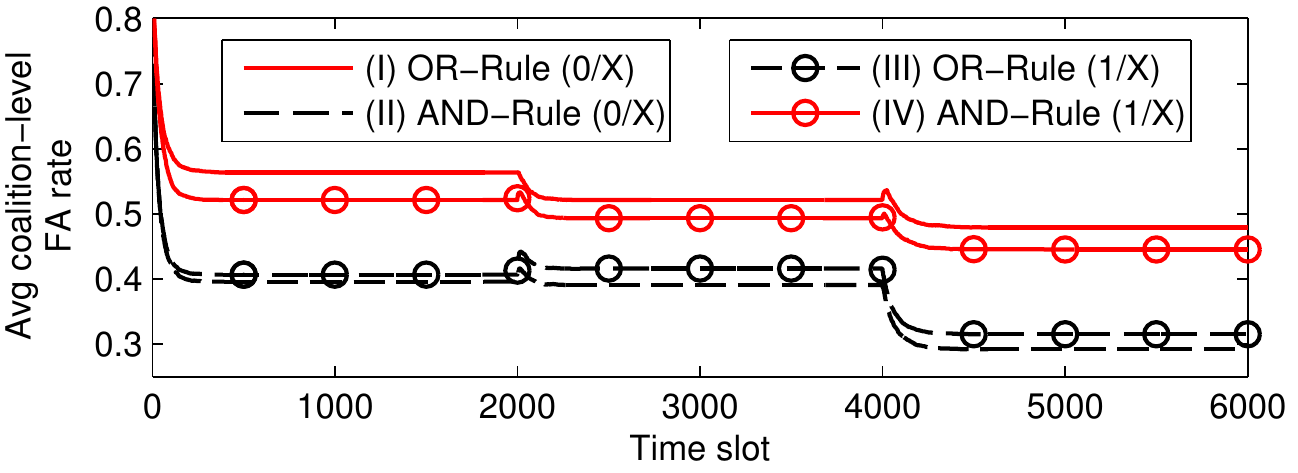}
    \caption{ Average coalition-level FA rate of the proposed hedonic game.}
    \label{Fig:CoalLevelFA}
\end{figure}

We also note that the changes in the numbers of PUs/SUs at time slots $2000$ and $4000$ do not cause significant disruptions to these games. Finally, while the overall system throughputs and energy consumption in Fig.~\ref{Fig:EnergyEfficiency}(a,d) are very similar under both MAC models, individual SUs report a higher degree of satisfaction under the 0/X-model as shown in Fig.~\ref{Fig:EnergyEfficiency}(b,c,e,f) since this model has a higher noncooperation cost (Section~\ref{SubSec:BottomDef}).

\ifdoublecolumn
\begin{figure*}
\normalsize
\begin{align}
U_\mathrm{1/X}^n(\eta ;\rho^n ) &=\beta ^n(1 - P_\mathrm{FA}^n(\eta ))\cdot\smashoperator[l]{\sum\limits_{\mathbf{\tilde{x}}  \in \{0,1\}^{|\rho^n | - 3}}}\Bigg\{{\sum_{(x_1,x_2)\in\{0,1\}^2}}\frac{|\eta|\Pr(X_{\xi_1}=x_1)\Pr(X_{\xi_2}=x_2)\prod _{i=1}^{|\rho^n|-3}\Pr(X_{\xi_{i+2}}=\tilde{x}_i)}{|\eta|+J_{\{\xi_1,\xi_2\}}((x_1,x_2))+J_{\rho^n\backslash\{\eta\}\backslash\{\xi_1\}\backslash\{\xi_2\}}(\mathbf{\tilde{x}})}\Bigg\}\label{Eq:NegExternalityBefore}\\
{U_\mathrm{1/X}^n(\eta ;\tilde \rho^n )} &=\beta ^n(1 - P_\mathrm{FA}^n(\eta ))\cdot\smashoperator[l]{\sum\limits_{\mathbf{\tilde{x}}  \in \{0,1\}^{|\rho^n | - 3}}}\Bigg\{{\sum_{x_1\in\{0,1\}}}\frac{|\eta|\Pr(X_{\xi_1\cup\xi_2}=x_1)\prod _{i=1}^{|\rho^n|-3}\Pr(X_{\xi_{i+2}}=\tilde{x}_i)}{|\eta|+J_{\{\xi_1\cup\xi_2\}}((x_1))+J_{\tilde\rho^n\backslash\{\eta\}\backslash\{\xi_1\cup\xi_2\} }(\mathbf{\tilde x})}\Bigg\}\label{Eq:NegExternalityAfter}\\
 U_\mathrm{1/X}^n(\eta ;\rho^n )&-U_\mathrm{1/X}^n(\eta ; \tilde\rho^n )=\beta ^n(1 - P_{{\mathrm{FA}}}^n(\eta ))\cdot\smashoperator[l]{\sum\limits_{\mathbf{\tilde{x}}  \in \{0,1\}^{|\rho^n | - 3}}} {\Bigg\{ \prod _{i=1}^{|\rho^n|-3}\Pr(X_{\xi_{i+2}}=\tilde{x}_i)\cdot A(\mathbf{\tilde{x}})  \Bigg\}} \label{Eq:NegExternalityProof}
\end{align}
\hrulefill
\vspace*{-10.2pt}
\label{Fig:Eq:NegExternalityBeforeAfter}
\end{figure*}
\fi

In Fig.~\ref{Fig:CoalLevelFA} we validate the superiority of the AND-rule over the OR-rule for heterogeneous environments (see Section~\ref{subsec:CoopSenseDef}) by comparing the average FA rates of the coalitions formed in the proposed game for the realistic scenario of Fig.~\ref{Fig:EnergyEfficiency}. Not surprisingly, we observe significant sensing accuracy degradation using the OR-rule, which can be explained by diverse PU-to-SU SNRs among the coalition members.

Next, we compare the computational complexity (\ref{Eq:Complexity}) and convergence time in Fig.~\ref{Fig:Complexity} for varying $N$ values, assuming stable SU/PU population and the 0/X MAC model. Note that when an SU determines whether it should switch to another channel, it needs the knowledge of the updated individual FA rates on the current and new channel for all involved SUs in both the proposed game and the game in \cite{HedonicSenseGame}. Thus, the complexity analysis in Section~\ref{SubSec:TopSol} applies to both games. We notice a significant reduction in convergence time and computational complexity when the proposed game is played, due primarily to provision for the social utility improvement in the preference relation (\ref{Eq:Preference})\footnote{Our game is still superior to \cite{HedonicSenseGame} in terms of the data rate, energy efficiency, and fairness when the social improvement requirement is removed.}. Moreover, extensive numerical experiments show that the actual number of switches needed for convergence is much smaller than the crude upper bound in Proposition~\ref{Prop:Converge}.

\begin{figure}[!t]
    \centering
    \includegraphics[scale=\FigScale]{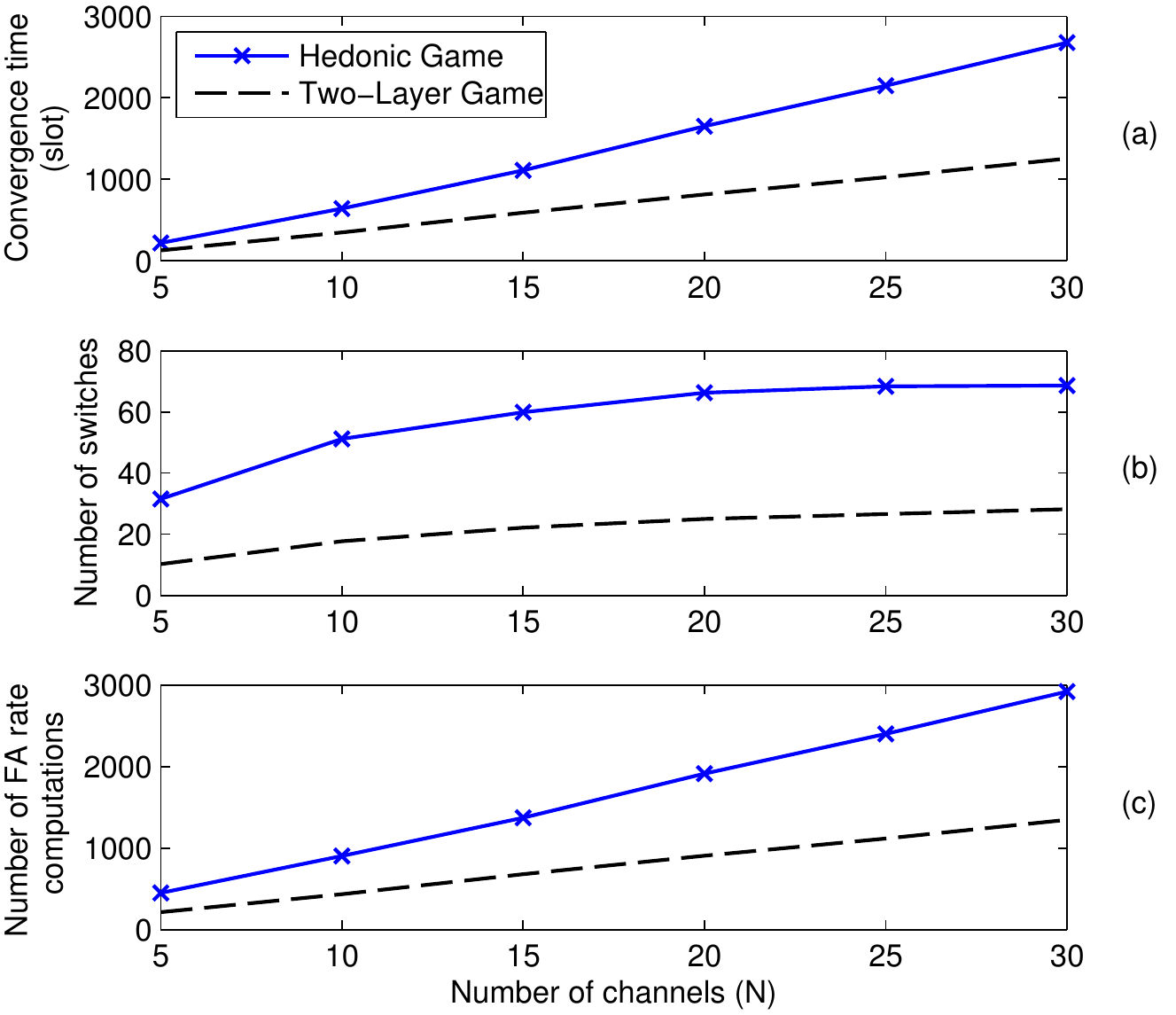}
    \caption{Computational complexity vs. $N$; $M=10$ SUs; 0/X-model: (a) convergence time $T_\textrm{Converge}$; (b) number of switches $|\mathcal{T}_\textrm{Switch}|$; (c) number of FA rate computations (\ref{Eq:Complexity}).}
    \label{Fig:Complexity}
\end{figure}

\begin{figure}[!t]
    \centering
    \includegraphics[scale=\FigScale]{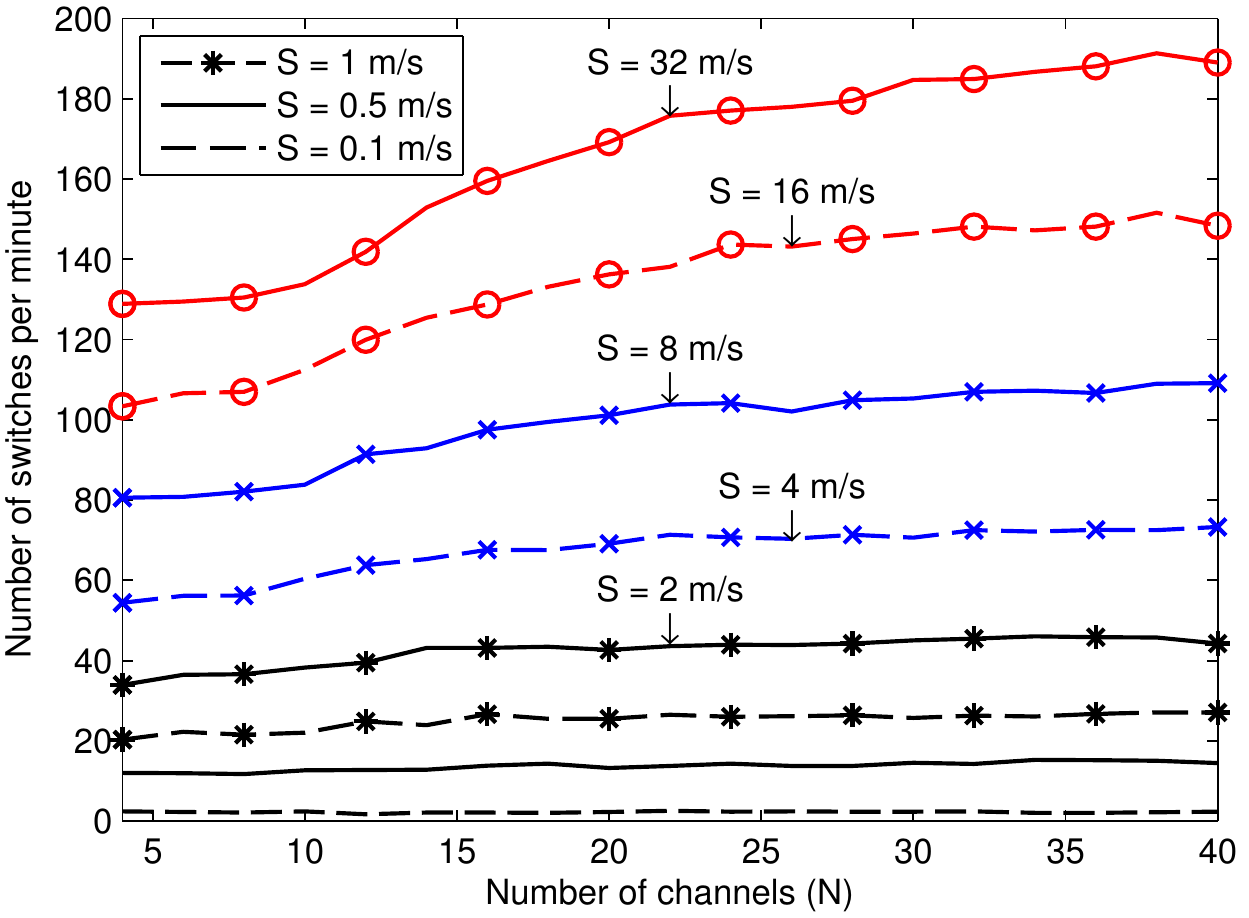}
    \caption{Switch frequency vs. number of channels ($N$) for various traveling speeds ($V$); 0/X-model; $M=10$ SUs.}
    \label{fig:MobilityNumOfSwitch}
\end{figure}

Finally, we study the effect of mobility on the performance of the proposed game. We assume that the nodes initially achieve a Nash-stable SU network partition and then move in randomly chosen independent directions at a constant speed $V$ following the random direction model \cite{MobilityModel}. When a node reaches the simulation region boundary, it moves in a new randomly chosen direction in the next time slot. Fig.~\ref{fig:MobilityNumOfSwitch} illustrates the \textit{switch frequency} given by the number of channel switches per minute in the network \cite{ZHanCoalSenseGame}. We observe that doubling of the mobile speed $V$ results in 10 to 40 additional channel switches per minute. However, for moderate speeds, the switch frequency is insensitive to the numbers of channels in the network. Frequent channel switching increases $|\mathcal{T}_\textrm{Switch}|$ in (\ref{Eq:Complexity}) and, thus, results in much larger  computational complexity. Moreover, the overhead grows with the mobile speed since the need for message exchange increases. While the proposed game adjusts well to sudden changes in the network settings (see Fig.~\ref{Fig:EnergyEfficiency}) and node positions \cite{YLuThesis}, continuous mobility, as in Fig.~\ref{fig:MobilityNumOfSwitch}, prevents convergence and increases complexity and overhead. Thus, alternative methods are required to resolve SU competition and provide high throughput in congested mobile CR networks. For example, in \cite{YLuTVT,YLuAllerton} we have investigated opportunistic sensing strategies that rely on multichannel diversity by adapting the reward of each SU pair to predicted local channel state information prior to sensing.

\section{Conclusion}\label{Sec:Conclusion}
The proposed two-layer game provides a comprehensive coalitional game-theoretical framework for cooperative sensing and access in multichannel multi-SU CR ad-hoc networks. Each SU is provided with the transmission opportunities it deserves and, thus, with sufficient incentives to participate in the proposed game.	Moreover, the contribution-based payoff allocation rule accelerates efficient distributed allocation of the SUs to the channels within the given CR spectrum. We also present a new cooperative sensing approach under MD constraints for guaranteed PU protection in the heterogeneous environment. Finally, we demonstrate that the proposed two-layer game outperforms the approaches in the literature in terms of energy efficiency, throughput, fairness, and is attractive in practical ad-hoc networks due to its relatively low computation load, distributed algorithm implementation, as well as robustness to changes in the number and positions of the SUs.

\ifdoublecolumn
\else
\begin{figure*}
\normalsize
\begin{align}
U_\mathrm{1/X}^n(\eta ;\rho^n ) &=\beta ^n(1 - P_\mathrm{FA}^n(\eta ))\cdot\smashoperator[l]{\sum\limits_{\mathbf{\tilde{x}}  \in \{0,1\}^{|\rho^n | - 3}}}\Bigg\{{\sum_{(x_1,x_2)\in\{0,1\}^2}}\frac{|\eta|\Pr(X_{\xi_1}=x_1)\Pr(X_{\xi_2}=x_2)\prod _{i=1}^{|\rho^n|-3}\Pr(X_{\xi_{i+2}}=\tilde{x}_i)}{|\eta|+J_{\{\xi_1,\xi_2\}}((x_1,x_2))+J_{\rho^n\backslash\{\eta\}\backslash\{\xi_1\}\backslash\{\xi_2\}}(\mathbf{\tilde{x}})}\Bigg\}\label{Eq:NegExternalityBefore}\\
{U_\mathrm{1/X}^n(\eta ;\tilde \rho^n )} &=\beta ^n(1 - P_\mathrm{FA}^n(\eta ))\cdot\smashoperator[l]{\sum\limits_{\mathbf{\tilde{x}}  \in \{0,1\}^{|\rho^n | - 3}}}\Bigg\{{\sum_{x_1\in\{0,1\}}}\frac{|\eta|\Pr(X_{\xi_1\cup\xi_2}=x_1)\prod _{i=1}^{|\rho^n|-3}\Pr(X_{\xi_{i+2}}=\tilde{x}_i)}{|\eta|+J_{\{\xi_1\cup\xi_2\}}((x_1))+J_{\tilde\rho^n\backslash\{\eta\}\backslash\{\xi_1\cup\xi_2\} }(\mathbf{\tilde x})}\Bigg\}\label{Eq:NegExternalityAfter}\\
 U_\mathrm{1/X}^n(\eta ;\rho^n )&-U_\mathrm{1/X}^n(\eta ;\tilde\rho^n)=\beta ^n(1 - P_{{\mathrm{FA}}}^n(\eta ))\cdot\smashoperator[l]{\sum\limits_{\mathbf{\tilde{x}}  \in \{0,1\}^{|\rho^n | - 3}}} {\Bigg\{ \prod _{i=1}^{|\rho^n|-3}\Pr(X_{\xi_{i+2}}=\tilde{x}_i)\cdot A(\mathbf{\tilde{x}})  \Bigg\}} \label{Eq:NegExternalityProof}
\end{align}
\hrulefill
\vspace*{-10.2pt}
\label{Fig:Eq:NegExternalityBeforeAfter}
\end{figure*}
\fi

\appendices 
\section{Proof of Proposition~\ref{Prop:NegExternality}}\label{Sec:NegExternalityProof}
Without loss of generality, we assume $\xi=\xi _1$ and $\zeta=\xi _2$ are the bottom-layer coalitions to be merged. The resulting bottom-layer partition is given by $\tilde\rho^n = \{ \xi _1  \cup \xi _2 \}  \cup (\rho^n \backslash \{ \xi _1 \} \backslash \{ \xi _2 \} )$ with $|\tilde\rho^n | = |\rho^n | - 1$. Expanding the bottom-layer coalition value ${U_\mathrm{1/X}^n(\eta ;\rho^n )}$ in (\ref{Eq:CoalValue1X}) over all possible values of $(X_{\xi_1},X_{\xi_2})$ in $\{0,1\}^2$ yields (\ref{Eq:NegExternalityBefore}) at the top of page \pageref{Eq:NegExternalityBefore} [cf. (\ref{Eq:X})]. Similarly, the bottom-layer coalition value ${U_\mathrm{1/X}^n(\eta ;\tilde\rho^n )}$ after the merger is derived in (\ref{Eq:NegExternalityAfter}) by expanding over all possible $X_{\xi_1\cup\xi_2}$ values. Noting that $\rho^n\backslash\{\eta\}\backslash\{\xi_1\}\backslash\{\xi_2\}=\tilde\rho^n\backslash\{\eta\}\backslash\{\xi_1\cup\xi_2\}$ and $P_{{\mathrm{FA}}}^n({\xi _1}\cup{\xi _1})=P_{{\mathrm{FA}}}^n({\xi _1})P_{{\mathrm{FA}}}^n({\xi _2})$ [cf. (\ref{Eq:CoalFAResultAND})], we subtract (\ref{Eq:NegExternalityAfter}) from (\ref{Eq:NegExternalityBefore}) and simplify to obtain (\ref{Eq:NegExternalityProof}) where [cf. (\ref{Eq:ParitionExcept})--(\ref{Eq:J})]
\begin{align}\label{Eq:A}
A(\mathbf{\tilde x})& =\Big(\frac{1}{b_1(\mathbf{\tilde x})}-\frac{1}{b_{12}(\mathbf{\tilde x})} \Big)|\eta|P_{{\mathrm{FA}}}^n({\xi _2})\Big(1 -P_{{\mathrm{FA}}}^n({\xi _1})\Big)\\
{}&+ \Big(\frac{1}{b_2(\mathbf{\tilde x})}- \frac{1}{b_{12}(\mathbf{\tilde x})} \Big)|\eta|P_{{\mathrm{FA}}}^n({\xi _1})\Big(1 - P_{{\mathrm{FA}}}^n({\xi _2})\Big)\notag
\end{align}
with
\begin{equation}
\begin{split}
{b_1}(\mathbf{\tilde x}) &= |\eta | + |{\xi _1}| + \sum\nolimits_{i = 1}^{|\rho^n | - 3} {(1 - {\tilde x}_i)|{\xi _{i {+} 2}}|} \\
{b_2}(\mathbf{\tilde x}) &= |\eta | + |{\xi _2}| + \sum\nolimits_{i = 1}^{|\rho^n | - 3} {(1 - {\tilde x}_i)|{\xi _{i {+} 2}}|} \\
{b_{12}}(\mathbf{\tilde x}) &= |\eta | + |{\xi _1}| + |{\xi _2}| + \sum\nolimits_{i = 1}^{|\rho^n | - 3} {(1- {\tilde x}_i)|{\xi _{i {+} 2}}|}.
\end{split}
\end{equation}
It is easy to show that $A(\mathbf{\tilde x}$ ) in (\ref{Eq:A}) and, thus, the expression in (\ref{Eq:NegExternalityProof}) must be nonnegative.

\bibliographystyle{IEEEtran}
\bibliography{TwoLayerGameBib}

\end{document}